\documentclass[conference, 10pt]{IEEEtran}
\IEEEoverridecommandlockouts
\usepackage{cite}
\usepackage{amsmath,amssymb,amsfonts,amsthm}
\usepackage{algorithmic}
\usepackage{graphicx}
\usepackage{textcomp}
\usepackage{xcolor}
\usepackage{threeparttable}
\usepackage{authblk}
\usepackage{color, colortbl} 
\usepackage{algorithmic}
\usepackage{algorithm}
\usepackage{subcaption}
\usepackage{float}
\usepackage{balance}
\usepackage[font=small]{caption}
\usepackage[font=footnotesize]{caption}%
\usepackage{hyperref}
\hypersetup{
pdftitle={Optimizing Unlicensed Band Spectrum Sharing With Subspace-Based Pareto Tracing},
pdfsubject={},
pdfauthor={Zachary J. Grey and Susanna Mosleh and Jacob Rezac and Yao Ma and Jason B. Coder and Andrew M. Dienstfrey},
pdfkeywords={LTE-LAA, Wi-Fi, wireless coexistence, MAC and physical layer parameters, active subspace, Pareto trace}
}
\usepackage{lipsum}

\definecolor{Gray}{gray}{0.9} 
\definecolor{LightCyan}{rgb}{0.88,1,1} 
\usepackage[first=0,last=9]{lcg}

\newtheorem{prop}{Proposition}

\def\BibTeX{{\rm B\kern-.05em{\sc i\kern-.025em b}\kern-.08em
    T\kern-.1667em\lower.7ex\hbox{E}\kern-.125emX}}

\begin{document}


\title{\huge Optimizing Unlicensed Band Spectrum Sharing With Subspace-Based Pareto Tracing
}

\author[$\dag$]{Zachary J. Grey\thanks{This work is U.S. Government work and not protected by U.S. copyright.}}
\author[$\S *$ ]{Susanna Mosleh}
\author[$\ddag$]{Jacob D. Rezac}
\author[$\ddag$]{Yao Ma}
\author[$\ddag$]{Jason B. Coder}
\author[$\dag$]{Andrew M. Dienstfrey}
\affil[$\dag$]{\footnotesize Applied and Computational Mathematics Division, Information Technology Lab, National Institute of Standards and Technology, Boulder, CO, 80305}
\affil[$\S$]{Associate, RF Technology Division, Communications Technology Lab, National Institute of Standards and Technology, Boulder, CO, 80305}
 \affil[$*$]{Department of Physics, University of Colorado, Boulder, CO, 80302}
\affil[$\ddag$]{\footnotesize RF Technology Division, Communications Technology Lab, National Institute of Standards and Technology, Boulder, CO, 80305}
\normalsize

\maketitle

\begin{abstract}
To meet the ever-growing demands of data throughput for forthcoming and deployed wireless networks, new wireless technologies like Long-Term Evolution License-Assisted Access (LTE-LAA) operate in shared and unlicensed bands. However, the LAA network must co-exist with incumbent IEEE 802.11 Wi-Fi systems. We consider a coexistence scenario where multiple LAA and Wi-Fi links share an unlicensed band. We aim to improve this coexistence by maximizing the key performance indicators (KPIs) of these networks simultaneously via dimension reduction and multi-criteria optimization. These KPIs are network throughputs as a function of medium access control protocols and physical layer parameters. We perform an exploratory analysis of coexistence behavior by approximating active subspaces to identify low-dimensional structure in the optimization criteria, i.e., few linear combinations of parameters for simultaneously maximizing KPIs. We leverage an aggregate low-dimensional subspace parametrized by approximated active subspaces of throughputs to facilitate multi-criteria optimization. The low-dimensional subspace approximations inform visualizations revealing convex KPIs over mixed active coordinates leading to an analytic Pareto trace of near-optimal solutions.


\end{abstract}


\section{Introduction}
As wireless communications evolve and proliferate into our daily lives, the demand for spectrum is growing dramatically. To accommodate this growth, wireless device protocols are beginning to transition from a predominantly-licensed spectrum to a shared approach in which use of the unlicensed spectrum bands appears to be inevitable. The main bottleneck of this approach is balancing new network paradigms with incumbent unlicensed networks, such as Wi-Fi.

As one strategy to manage spectrum scarcity, providers are beginning to operate Long-Term Evolution License-Assisted Access (LTE-LAA) in unlicensed bands\footnote{Our focus is the operation of LTE base stations in an unlicensed band. However, these base stations may have permission to utilize a licensed band as well.}. Even though operating LAA in unlicensed bands improves spectral-usage efficiency, it may have an enormous influence on Wi-Fi operation and create a number of challenges for both Wi-Fi and LTE networks as a means of constructively sharing the spectrum. Understanding and addressing these challenges calls for a deep dive into the operations and parameter selection of both networks in the medium access control (MAC) and physical (PHY) layers.

There have been many investigations of fairness in spectrum sharing among LAA and Wi-Fi networks \cite{Quek-Access2016,Cano-Acm2017,Mehrnoush-2018}---although, these works do not consider optimizing key performance indicators (KPIs). Contrary to \cite{Quek-Access2016,Cano-Acm2017,Mehrnoush-2018}, the authors in \cite{Gao-2017} and \cite{Gao-2019} maximize LAA throughput and total network sum rate, respectively, over contention window sizes of both networks while guaranteeing the Wi-Fi throughput satisfies a threshold. However \cite{Gao-2017} and \cite{Gao-2019} optimize only a single MAC layer parameter. A multi-criteria optimization problem was formulated in \cite{Yin-2016} to satisfy the quality of service requirements of LAA eNodeBs by investigating the trade-off between the co-channel interference in the licensed band and the Wi-Fi collision probability in the unlicensed band. However, maximizing the Wi-Fi throughput was omitted. Considering both PHY and MAC layer parameters, \cite{Mosleh-GC-2019} maximizes the weighted sum rate of an LAA network subject to Wi-Fi throughput constraint with respect to
the fraction of time that LAA is active. Alternatively, advantageous sharing of spectrum can be modeled as a multi-criteria optimization problem where both Wi-Fi and LAA KPIs, such as network throughputs on the unlicensed bands, are simultaneously maximized with respect to their MAC and PHY layer parameters. The set of maximizing arguments quantify the inherent trade-off between LAA and Wi-Fi throughputs.

The multi-criteria optimization formalism we propose is further complicated by the high dimensionality of the input space. Specifically, our model requires $17$ MAC and PHY parameters to characterize the coexistence performance. Previous experience suggests that not all parameter combinations are equally important in determining KPIs quality. Active subspaces supplement an exploratory approach for determining parameter combinations which change KPI values the most, on average. The sets of parameter combinations defined by the active subspaces help inform KPI approximations and visualizations over an aggregate low-dimension subspace---simplifying the multi-criteria optimization. 

We incorporate active subspace dimension reduction into a multi-criteria optimization framework to analyze the shared spectrum coexistence problem involving Wi-Fi and LAA. The proposed technique is extensible to many other spectrum sharing and communication systems, but LTE-LAA is used as an example. The dimension reduction supplements a trade-off analysis of network throughputs by computing a Pareto trace. The Pareto trace provides a \emph{continuous} approximation of Pareto optimal points in a common domain of a multi-criteria problem \cite{boyd2004convex, bolten2020tracing}---resulting in a near-best trade-off between differing throughputs. This offers a continuous description of a parameter subset which quantifies high quality performance of both networks, facilitated by a dimension reduction.


\section{System Model and Assumptions} \label{sec:sys_model}
We consider a downlink coexistence scenario where two mobile network operators (MNOs) operate over the same shared unlicensed industrial, scientific, and medical radio band. Previously, unlicensed bands were dominated by Wi-Fi traffic and, occasionally, used by commercial cellular carriers
for offloading data otherwise communicated via
LTE in the licensed spectrum. Lately, LTE carriers are choosing to operate in unlicensed bands in addition to data offloading. We assume the MNOs use time sharing to simultaneously operate in this band and we aim to analyze competing trade-offs in throughputs of the Wi-Fi and LTE systems. The network throughput is a function of both physical and MAC layer parameters. We introduce the parameters defining the network topology, the PHY layer, MAC
layer protocols, and briefly discuss the relation of these variables to network throughput.

We consider a coexistence scenario in which the LAA network consists of $n_{L}$ eNodeBs, while the Wi-Fi network is composed of $n_{W}$ access points (APs)\footnote{We are primarily focused on the operation of cellular base stations in the unlicensed bands. However, LTE base stations may have permission to utilize a licensed band as well.}. The eNodeBs and APs are randomly distributed over a particular area, while LAA user equipment (UEs) and Wi-Fi clients/stations (STAs) are uniformly and independently distributed around each eNodeB and AP, respectively. Each transmission node serves a set of single antenna UEs/STAs and the user association is based on the received power. We assume \textit{(i)} both Wi-Fi and LAA are in the saturated traffic condition, i.e., at least one packet is waiting to be sent, \textit{(ii)} there are neither hidden nodes nor false alarm/miss detection problems in the network\footnote{We assume perfect spectrum sensing in both systems. The impact of imperfect sensing is beyond the scope of this paper and investigating the effect of sensing errors is an important topic for future work.}, and \textit{(iii)} the channel knowledge is ideal, so, the only source of unsuccessful transmission is collision. The physical data rate of the LAA and Wi-Fi networks is a function of signal-to-interference-plus-noise ratio (SINR) that is related to and changes with the link distances and propagation model. Any changes in data rates lead to different network throughput.

The medium access key feature in both Wi-Fi and LAA involves the station accessing the medium to sense the channel by performing clear channel assessment prior to transmitting. The station only transmits if the medium is determined to be idle. Otherwise, the transmitting station refrains from transmitting data until it senses the channel is available. Although LAA and Wi-Fi technologies follow similar channel access procedures, they utilize different carrier sense schemes, different channel sensing threshold levels, and different channel contention parameters, leading to different unlicensed channel access probabilities and thus, different throughputs.


Conforming with the analytical model in \cite{Mosleh-GC-2019, Mosleh-VTC2020,Mosleh-ICC2020}, the LAA and Wi-Fi throughputs, indicated respectively by $S_{\mathcal{L}}$ and $S_{\mathcal{W}}$, are functions of $m$ MAC and PHY layer parameters in a vector $\boldsymbol{\theta}$ conditioned on fixed values in a vector $\boldsymbol{x}$,
\begin{align} \label{eq:throughputs}
\nonumber
S_{\mathcal{L}}&: \mathbb{R}^m \times \lbrace \boldsymbol{x} \rbrace \rightarrow \mathbb{R}:(\boldsymbol{\theta}, \boldsymbol{x}) \mapsto S_{\mathcal{L}}(\boldsymbol{\theta}; \boldsymbol{x}),\\ 
S_{\mathcal{W}}&: \mathbb{R}^m\times \lbrace \boldsymbol{x} \rbrace \rightarrow \mathbb{R}:(\boldsymbol{\theta}, \boldsymbol{x}) \mapsto S_{\mathcal{W}}(\boldsymbol{\theta}; \boldsymbol{x}).
\end{align}
For this application, LAA throughput $S_{\mathcal{L}}$ and Wi-Fi throughput $S_{\mathcal{W}}$ are only considered functions of $m$ variable parameters in $\boldsymbol{\theta}$. This numerical study considers $m=17$ parameters summarized in Table \ref{Tb1}. 
To simplify this study, we fix the remaining parameters in $\boldsymbol{x}$ governing the majority of the physical characteristics of the communication network---constituting a fixed \emph{scenario} $\boldsymbol{x}$ for a parameter study over $\boldsymbol{\theta}$'s. 

The problem of interest is to maximize a convex combination of network throughputs for the fixed scenario $\boldsymbol{x}$ over the MAC and PHY parameters $\boldsymbol{\theta}$ in a multi-criteria optimization. Mathematically, we define the \emph{Pareto front} by the following optimization problem:
\begin{equation} \label{eq:multiopt}
    \underset{\boldsymbol{\theta} \in \mathcal{D}  \subset \mathbb{R}^m}{\text{maximize}} \,\,tS_{\mathcal{L}}(\boldsymbol{\theta}; \boldsymbol{x}) + (1-t)S_{\mathcal{W}}(\boldsymbol{\theta}; \boldsymbol{x}),
\end{equation}
for all $t \in [0,1]$ where $\mathcal{D}$ is the parameter domain defined by the ranges in Table \ref{Tb1}. The goal is to quantify a smooth trajectory $\boldsymbol{\theta}(t)$ through MAC and PHY parameter space, or \emph{trace} \cite{bolten2020tracing}, such that the convex combination of throughputs is maximized over a map $\boldsymbol{\theta}:[0,1] \rightarrow \mathcal{D}$. In Section \ref{sec:pareto_trace} we formalize this notion of a trace. In Section \ref{sec:AS}, we summarize an exploratory approach for understanding to what extent problem \eqref{eq:multiopt} is convex \cite{boyd2004convex} and how we can intuitively regularize. The empirical evidence generated through visualization and dimension reduction provide justification for convex quadratic approximations and subsequent quadratic trace in Section \ref{sec:Numerics}. 
\begin{table*}[t]
\captionsetup{font=footnotesize}
 \caption{MAC and PHY parameters influencing throughputs }\label{Tb1}
 \begin{threeparttable}
 \centering
 \scalebox{0.98}{
 \scriptsize
  \begin{tabular}{c|p{3.4cm} |c|c}
    \hline
    Parameters & Description &Bounds & Nominal \\ \hline
    \rowcolor{LightCyan}
     $\theta_1$ & Wi-Fi min contention window size&(8, 1024)& 516\\[0.1cm] 
      \rowcolor{LightCyan}
  $\theta_2$ &  LAA min contention window size&(8, 1024)& 516\\[0.1cm] 
   \rowcolor{LightCyan}
    $\theta_3$ & Wi-Fi max back-off stage** &($10^{-4}$, 8)& 4\\[0.1cm] 
     \rowcolor{LightCyan}
    $\theta_4$ & LAA max back-off stage**&($10^{-4}$, 8)& 4\\[0.1cm] 
    \rowcolor{LightCyan}
    $\theta_{5}$ & Distance between transmitters &(10 m, 20 m)& 15 m\\[0.1cm]
    \rowcolor{LightCyan}
    $\theta_{6}$ & Minimum distance between transmitters and receivers&(10m, 35 m)& 22.5 m\\[0.1cm]
    \rowcolor{LightCyan}
    $\theta_{7}$ & Height of each LAA eNodeB and Wi-Fi AP&(3 m, 6 m) & 4.5 m\\[0.1cm]
    \rowcolor{LightCyan}
    $\theta_{8}$ & Height of each LAA UEs and Wi-Fi STAs & (1 m, 1.5 m) & 1.25 m\\[0.1cm]
        \rowcolor{LightCyan}
    $\theta_9$ & Standard deviation of shadow fading& (8.03, 8.29) & 8.16 \\[0.1cm]
    \rowcolor{LightCyan}
    $\theta_{10}$ & $k_{\text{LOS}}^{\star}$
    & (45.12, 46.38)& 45.75\\[0.1cm]
    \rowcolor{LightCyan}
    $\theta_{11}$ & $k_{\text{NLOS}}^{\star}$
    & (34.70, 46.38) & 40.54\\[0.1cm]
  \end{tabular} \hspace{0.12in}
  \begin{tabular}{c|p{3.4cm} |c|c}
    \hline
    Parameters & Description &Bounds & Nominal \\ \hline
    \rowcolor{LightCyan}
    $\theta_{12}$ &  $\alpha_{\text{LoS}}^{\star}$ 
    & (17.3, 21.5) & 19.4 \\[0.1cm]
    \rowcolor{LightCyan}
    $\theta_{13}$ &  $\alpha_{\text{NLoS}}^{\star}$
    &(31.9, 38.3) & 35.1\\[0.1cm] 
    \rowcolor{LightCyan}
    $\theta_{14}$ & Transmitter antenna gain** &($10^{-4}$ dBi, 5 dBi) & 2.5 dBi\\[0.1cm]
    \rowcolor{LightCyan}
    $\theta_{15}$ & Noise figure at each receiver&(5 dB, 9 dB) & 7 dB\\[0.1cm]
        \rowcolor{LightCyan}
    $\theta_{16}$ & Transmit power at each LAA eNodeB and Wi-Fi AP&(18 dBm, 23 dBm) & 20.5 dBm\\[0.1cm]
    \rowcolor{LightCyan}
    $\theta_{17}$ & Carrier channel bandwidth &(10 MHz, 20 MHz) & 15 MHz\\ [0.1cm]
    \rowcolor{Gray}
    $x_1$ & Number of LAA eNodeBs ($n_L$) &--& 6\\[0.1cm] 
    \rowcolor{Gray}
    $x_2$ & Number of Wi-Fi APs ($n_W$) &--& 6\\[0.1cm] 
    \rowcolor{Gray}
    $x_3$ & Number of LAA UEs &--& 6\\[0.1cm] 
    \rowcolor{Gray}
    $x_4$ & Number of Wi-Fi STAs &--&6\\[0.1cm] 
    \rowcolor{Gray}
    $x_5$ & Number of unlicensed channels &--& 1\\[0.1cm] 
    \rowcolor{Gray}
    $x_{8}$ &  Scenario width &--&120 m\\[0.1cm]
    \rowcolor{Gray}
    $x_{9}$ &  Scenario height &--&80 m
    \end{tabular}}
     \begin{tablenotes}
       \scriptsize
           \item \text{Note: parameter ranges are established by 3GPP TS 36.213 V15.6.0 and 3GPP TR. 36.889 v13.0.0.}
      \item  \text{${}^{\star}$The path-loss for both line-of-sight (LoS) and non-LoS scenarios can be computed as $k + \alpha\log_{10}(d)$ in dB, where $d$ is the distance in meters between the transmitter }
      \item \text{and the receiver. **Note that typical lower bounds are taken as zero however we transform parameters to a log-space and supplement with a sufficiently small lower bound.}
    \end{tablenotes}
   \end{threeparttable}
\end{table*}
\normalsize

\section{Pareto Tracing} \label{sec:pareto_trace}
We refer to \eqref{eq:multiopt} as a maximization of the convex total objective or \emph{scalarization}. In this case, we have a single degree of freedom to manipulate the scalarization parametrized by $t \in [0,1]$ such that
$
J_{t}(\boldsymbol{\theta}) = (1- t) S_{\mathcal{W}}(\boldsymbol{\theta}; \boldsymbol{x}) + t S_{\mathcal{L}}(\boldsymbol{\theta}; \boldsymbol{x})
$.
By virtue of the necessary conditions for a (locally) Pareto optimal solution, we must determine $\boldsymbol{\theta} \in \mathbb{R}^m$ critical for $J_t$ which necessarily implies $\nabla J_t(\boldsymbol{\theta}) = (1- t)\nabla S_{\mathcal{W}}(\boldsymbol{\theta}) + t\nabla S_{\mathcal{L}}(\boldsymbol{\theta}) = \boldsymbol{0}$ where $\boldsymbol{0} \in \mathbb{R}^m$ is a vector of zeros and $\nabla$ is the gradient with respect to $\boldsymbol{\theta}$---this is referred to as the \emph{stationarity} condition. Moreover, denoting the Hessian matrices $\nabla^2 S_{\mathcal{L}}(\boldsymbol{\theta}),\, \nabla^2 S_{\mathcal{W}}(\boldsymbol{\theta}),\,\nabla^2 J_t(\boldsymbol{\theta}) \in \mathbb{R}^{m \times m}$, $\boldsymbol{\theta}$ satisfies strict second order $J_t$-optimality and is a locally (unique) Pareto optimal solution if $\nabla^2 J_t(\boldsymbol{\theta}) = (1- t) \nabla^2 S_{\mathcal{W}}(\boldsymbol{\theta}; \boldsymbol{x}) + t \nabla^2 S_{\mathcal{L}}(\boldsymbol{\theta}; \boldsymbol{x})$ is (symmetric) negative definite \cite{bolten2020tracing, boyd2004convex}.

First, we review the necessary conditions for a continuous (in $t$) solution to \eqref{eq:multiopt}. Provided the set of all Pareto optimal solutions is convex, we can continuously parametrize the set of all Pareto optimal solutions to \eqref{eq:multiopt} as $t \mapsto \boldsymbol{\theta}(t)$ for all $t \in [0,1]$ considering $t$ as pseudo-time for an analogous trajectory through parameter space---launching from one minimizing argument to another.

\begin{prop}\label{prop:ODE}
	Given full rank $\nabla^2   J_t(\boldsymbol{\theta}(t)) \in \mathbb{R}^{m \times m}$, the one-dimensional immersed submanifold parametrized by $\boldsymbol{\theta}(t) \in \mathbb{R}^m$ for all $t \in [0,1]$ is necessarily Pareto optimal such that
	$$
	\nabla^2 J_t(\boldsymbol{\theta}(t))\dot{\boldsymbol{\theta}}(t) = \nabla S_{\mathcal{W}}(\boldsymbol{\theta}(t)) - \nabla  S_{\mathcal{L}}(\boldsymbol{\theta}(t)).
	$$
\end{prop}

\begin{proof}
	Differentiating the stationarity condition by composing in pseudo-time, $\nabla J_t \circ \boldsymbol{\theta}(t) = \boldsymbol{0}$, results in
	\small
\begin{align*}
\boldsymbol{0}&{}=\frac{d}{dt}\left(\nabla J_t \circ \boldsymbol{\theta}(t)\right)\\
&{}= \frac{d}{dt} (1- t)\left(\nabla S_{\mathcal{W}} \circ \boldsymbol{\theta}(t) \right) + \frac{d}{dt}t\left(\nabla S_{\mathcal{L}}\circ \boldsymbol{\theta}(t)\right)\\
&{}=-\nabla S_{\mathcal{W}} \circ \boldsymbol{\theta}(t) + (1-t)\left(\nabla^2_{\boldsymbol{\theta}} S_{\mathcal{W}} \circ \boldsymbol{\theta}(t)\right)\dot{\boldsymbol{\theta}}(t)\\
&\hspace{1cm} +\nabla S_{\mathcal{L}}\circ \boldsymbol{\theta}(t) + t\left(\nabla^2_{\boldsymbol{\theta}}S_{\mathcal{L}} \circ \boldsymbol{\theta}(t)\right)\dot{\boldsymbol{\theta}}(t)\\
&{}=\nabla^2 J_t(\boldsymbol{\theta}(t))\dot{\boldsymbol{\theta}}(t) - \left(\nabla S_{\mathcal{W}} \circ \boldsymbol{\theta}(t) - \nabla S_{\mathcal{L}} \circ \boldsymbol{\theta}(t)\right).
\end{align*}
\normalsize

Then, the flowout along necessary Pareto optimal solutions constitutes an immersed submanifold of $\mathbb{R}^m$ nowhere tangent to the integral curve generated by the system of differential equations (See \cite{Lee2003}, Thm. 9.20)---i.e., assuming $\nabla^2 J_t(\boldsymbol{\theta}(t))$ is full rank, $\nabla^2   J_t(\boldsymbol{\theta}(t))^{-1}\left(\nabla   S_{\mathcal{W}}(\boldsymbol{\theta}(t)) - \nabla   S_{\mathcal{L}}(\boldsymbol{\theta}(t))\right)$ is the infintesimal generator of a submanifold $\mathcal{P} \subseteq \mathbb{R}^m$ of locally Pareto optimal solutions contained in the flowout.
\end{proof}

We note that the system of equations in Prop. \ref{prop:ODE}, proposed in \cite{bolten2020tracing}, constitutes a set of \textit{necessary} conditions for optimality. The utility of Prop. \ref{prop:ODE} offers an interpretation that the solution set (if it exists) constitutes elements of a submanifold in $\mathbb{R}^m$. This formalism establishes a theoretical foundation for the use of manifold learning or splines over sets of points which are approximately Pareto optimal.

Suppose $S_{\mathcal{W}}$ and $S_{\mathcal{L}}$ are well approximated by convex quadratics as \textit{surrogates}, 
\begin{align*}
-S_{\mathcal{L}}(\boldsymbol{\theta}; \boldsymbol{x})&\approx \boldsymbol{\theta}^T \boldsymbol{Q}_{\mathcal{L}} \boldsymbol{\theta} + \boldsymbol{a}_{\mathcal{L}}^T\boldsymbol{\theta} + c_{\mathcal{L}} \\\nonumber
-S_{\mathcal{W}}(\boldsymbol{\theta}; \boldsymbol{x})&\approx \boldsymbol{\theta}^T \boldsymbol{Q}_{\mathcal{W}} \boldsymbol{\theta} + \boldsymbol{a}_{\mathcal{W}}^T\boldsymbol{\theta} + c_{\mathcal{W}},
\end{align*}
such that $\boldsymbol{Q}_{\mathcal{L}}, \, \boldsymbol{Q}_{\mathcal{W}} \in S_{++}^m$,  $\boldsymbol{a}_{\mathcal{L}},\, \boldsymbol{a}_{\mathcal{W}} \in \mathbb{R}^m$, and $c_{\mathcal{L}},\,c_{\mathcal{W}} \in \mathbb{R}$ where $S_{++}^m$ denotes the collection of $m$-by-$m$ postive definite matrices---note the change in sign convention. Consequently, applying the sationarity condition to the form of the quadratic approximations results in the closed form solution of the Pareto trace,
\begin{equation} \label{eq:quad_trace}
    \boldsymbol{\theta}(t) = \frac{1}{2}\left[t \boldsymbol{Q}_{\mathcal{L}} + (1-t)\boldsymbol{Q}_{\mathcal{W}}\right]^{-1}\left[(t-1) \boldsymbol{a}_{\mathcal{W}} - t\boldsymbol{a}_{\mathcal{L}}\right],
\end{equation}
referred to in this context as a \textit{quadratic trace}. The quadratic trace is derived from the stationarity condition and thus consistent with Prop. \ref{prop:ODE}. The challenge in our context with otherwise unknown forms of $S_{\mathcal{W}}$ and $S_{\mathcal{L}}$ is: \emph{how can we assess conditions informing \eqref{eq:quad_trace} or regularize the solve to guarantee these conditions?} We offer an approach which assess these conditions by visualization and subsequent regularization through \textit{subspace-based dimension reduction} to inform convex quadratic surrogates as approximations satisfying Prop. \ref{prop:ODE}.

\section{Active Subspaces} \label{sec:AS}
Following the development in \cite{Constantine2015}, we introduce an exploratory method for simplifying the problem statement in \eqref{eq:multiopt}. For ease of exposition, in this section we denote either scalar-valued throughput function by $S:\mathcal{D} \subset \mathbb{R}^m \rightarrow \mathbb{R}$ with compact domain $\mathcal{D}$ and assume all integrals and derivatives used in this discussion exist as measurable functions. The main results of the section rely on an eigendecomposition of the symmetric positive semi-definite matrix $\boldsymbol{C}\in\mathbb{R}^{m\times m}$ defined as
\begin{equation} \label{eq:C}
\boldsymbol{C} = \int_{\mathcal{D}} \nabla S(\boldsymbol{\theta}) \nabla S(\boldsymbol{\theta})^T d\boldsymbol{\theta}
\end{equation}
with entries $C_{ij} = \int_\mathcal{D} \left(\partial S/\partial \theta_i\right)\vert_{\boldsymbol{\theta}}\left(\partial S/\partial \theta_j\right)\vert_{\boldsymbol{\theta}} d\boldsymbol{\theta}$ for $i,j=1,\ldots,m$. For this application, the integral is taken uniformly over $\mathcal{D}$. The compact domain  $\mathcal{D}$ is a hyper-rectangle constructed\footnote{The chosen definition of \eqref{eq:C} weights all parameter combinations equally over $\mathcal{D}$ and restricts the integration to feasible values (as summarized in Table \ref{Tb1}). Alternative choices are available in scenarios where it is more appropriate to weight parameters differently, but uniform is suitable to our application.} from the Cartesian product of lower and upper bounds, $\theta_{i,\ell} \leq \theta_i \leq \theta_{i,u}$ for all $i=1,\dots,m$.
\subsection{Interpretability of Active Subspaces}
If $\text{rank}(\boldsymbol{C}) = r < m$, its eigendecomposition $\boldsymbol{C} = \boldsymbol{W}\boldsymbol{\Lambda} \boldsymbol{W}^T$ with orthogonal $\boldsymbol{W}$ satisfies $\boldsymbol{\Lambda} = \mathrm{diag}(\lambda_1,\dots,\lambda_m)$ with 
\begin{equation} \label{eq:eig_decay}
    \lambda_1 \geq \lambda_2 \geq \dots \geq \lambda_r > \lambda_{r+1} = \dots = \lambda_m = 0.
\end{equation}
This defines two sets of important
$
\boldsymbol{W}_r = [\boldsymbol{w}_1 \dots \boldsymbol{w}_r] \in \mathbb{R}^{m \times r}
$
and unimportant
$
\boldsymbol{W}_r^{\perp} = [\boldsymbol{w}_{r+1} \dots \boldsymbol{w}_m] \in \mathbb{R}^{m \times (m-r)}
$
\emph{directions} over the domain. The column span of $\boldsymbol{W}_r$ and $\boldsymbol{W}_r^{\perp}$ constitute the \emph{active} and \emph{inactive} subspaces, respectively. Note that \eqref{eq:C} depends on a single scalar-valued response and potentially differs for the separate throughputs in \eqref{eq:throughputs}.

What do we mean by \emph{important directions}? Partitioning $\boldsymbol{W}$ into $m$ orthonormal eigenvectors $\boldsymbol{w}_i \in \mathbb{R}^m$ representing the columns, $\boldsymbol{W} = [\boldsymbol{w}_1 \dots \boldsymbol{w}_m]$, we can simplify $\boldsymbol{w}_i^T\boldsymbol{C}\boldsymbol{w}_i$ to obtain an expression for the eigenvalues,
\begin{equation}\label{eq:eigs}
    \lambda_i = \int_{\mathcal{D}}\left(\boldsymbol{w}_i^T \nabla S(\boldsymbol{\theta})\right)^2 d\boldsymbol{\theta}.
\end{equation}
Reinterpreting the integral by definition of the expectation, $\mathbb{E}[f(\boldsymbol{\theta})] = \int_{\mathcal{D}}f(\boldsymbol{\theta}) d\boldsymbol{\theta}$ for any measurable $f$, the eigenvalues can be interpreted as the mean squared directional derivative of $S$ in the direction of $\boldsymbol{w}_i \in \mathbb{R}^m$. Precisely, 
the directional derivative can be written $dS_{\boldsymbol{\theta}}[\boldsymbol{w}] = \boldsymbol{w}^T \nabla S(\boldsymbol{\theta})$
and we obtain $\lambda_i = \mathbb{E}[dS^2_{\boldsymbol{\theta}}[\boldsymbol{w}_i]]$---the $i^{th}$ eigenvalue is the mean squared directional derivative in the direction of the $i^{th}$ eigenvector.
Thus, the inherent ordering \eqref{eq:eig_decay} of the eigenpairs $\lbrace(\lambda_i, \boldsymbol{w}_i)\rbrace_{i=1}^m$ indicate directions $\boldsymbol{w}_i$ which change the function $S$ the most, on average, up to the $r+1,\dots,m$ directions which \emph{do not change the function at all} \cite{Constantine2015}. In other words, the directional derivatives over the inactive subspace, $\mathrm{Range}(\boldsymbol{W}_r^{\perp})$, are zero provided the corresponding eigenvalues are zero. In fact, either throughput response from \eqref{eq:throughputs} is referred to as a \emph{ridge function} over $\boldsymbol{\theta}$'s if and only if $dS_{\boldsymbol{\theta}}[\boldsymbol{w}] = 0$ for all $\boldsymbol{w} \in \text{Null}(\boldsymbol{W}_r^T)$.  

\subsection{Ridge Approximations}
Naturally, if the trailing eigenvalues are merely small as opposed to identically zero, then the function changes much less over the inactive directions which have smaller directional derivatives. This lends itself to a framework for reduced-dimension approximation of the function such that we only approximate changes in the function over the first $r$ active directions and take the approximation to be constant over the trailing $m-r$ inactive directions \cite{Constantine2015}. Such an approximation to $S$ is called a \emph{ridge approximation} by a function $H$ referred to as the \emph{ridge profile}, i.e.,
\begin{equation}\label{eq:ridge_approx}
    S(\boldsymbol{\theta}) \approx H(\boldsymbol{W}_r^T \boldsymbol{\theta}).
\end{equation}
In the event that the trailing eigenvalues of $\boldsymbol{C}$ are zero, then the approximation is exact for a particular $H$ \cite{Constantine2015}. 

In either case, approximation or an exact ridge profile, the possibility of reducing dimension by projection to fewer, $r< m$, \emph{active coordinates} $\boldsymbol{y} = \boldsymbol{W}_r^T\boldsymbol{\theta} \in \mathbb{R}^r$ can enable higher-order polynomial approximations for a given data set of coordinate-output pairs and an ability to \emph{visualize} the approximation. For example, we can visualize the approximation by projection to the active coordinates when $r$ is chosen to be $1$ or $2$ based on the decay and gaps in the eigenvalues. These subsequent visualizations are referred to as \emph{shadow plots} \cite{Grey2017} or graphs $\lbrace (\boldsymbol{W}_r^T\boldsymbol{\theta}_i, S(\boldsymbol{\theta}_i)\rbrace_{i=1}^N$ for $N$ samples $\boldsymbol{\theta}_i$ drawn uniformly (for this application). A strong decay leading to a small sum of trailing eigenvalues implies an improved approximation over relatively few important directions while larger gaps in eigenvalues imply an improved approximation to the low-dimensional subspace \cite{Constantine2015}. Identifying if this structure exists depends on the decay and gaps in eigenvalues. We can subsequently exploit any reduced dimensional visualization and approximation to simplify our problem \eqref{eq:multiopt}. However, we must reconcile that our problem of interest involves two separate computations of throughput, $S_{\mathcal{L}}$ and $S_{\mathcal{W}}$.

\subsection{Mixing Disparate Subspaces}
Independently approximating active subspaces for the objectives $S_{\mathcal{L}}$ and $S_{\mathcal{W}}$ generally results in different subspaces of the shared parameter domain. The next challenge is to define a common subspace that, while sub-optimal for each objective, is nevertheless sufficient to capture the variability of both simultaneously. Assume that we can reduce important parameter combinations to a common dimension $r$ of potentially distinct subspaces. These subspaces are spanned by the column spaces of $\boldsymbol{W}_{r,\mathcal{L}}$ and $\boldsymbol{W}_{r, \mathcal{W}}$ chosen as the first $r$ eigenvectors resulting from separate approximations of \eqref{eq:C} for LAA and Wi-Fi throughputs, respectively. The challenge is to appropriately ``mix'' the subspaces so we may formulate a solution to \eqref{eq:multiopt} over a common dimension reduction. 

One method to find an appropriate subspace mix is to take the union of both subspaces. However, if $r\geq 2$ and $\text{Range}(\boldsymbol{W}_{r,\mathcal{L}}) \cap \text{Range}(\boldsymbol{W}_{r, \mathcal{W}}) = \lbrace \boldsymbol{0} \rbrace$ then the combined subspace dimension is inflated and  visualization of subsequent convex approximations becomes challenging. We use interpolation between the two subspaces to overcome these difficulties and retain the common reduction to an $r$-dimensional subspace. The space of all $r$-dimensional subspaces in $\mathbb{R}^m$ is the $r(m-r)$-dimension Grassmann manifold (Grassmannian\footnote{Formally, an element of the Grassmannian is an equivalence class, $[\boldsymbol{U}_r]$, of all orthogonal matrices whose first $r$ columns span the same subspace as $\boldsymbol{U}_r \in \mathbb{R}^{m \times r}$. That is, the equivalence relation $X \sim Y$ is given by $\text{Range}(X) =\text{Range}(Y)$ denoted $[X]$ or $[Y]$ for $X,Y \in \mathbb{R}^{m \times r}$ full rank with orthonormal columns.}) denoted $\text{Gr}(r,m)$ \cite{edelman1998geometry}. Utilizing the analytic form of a geodesic over the Grassmannian \cite{edelman1998geometry}, we can smoothly interpolate from one subspace to another---an interpolation which is, in general, non-linear. This is particularly useful because the \emph{distance} between any two subspaces along such a path, 
$
[\boldsymbol{U}_r]:\mathbb{R} \rightarrow \text{Gr}(r, m):s \mapsto [\boldsymbol{U}_r(s)] \,\, \text{for all}\,\, s \in [0,1],
$
is minimized between the two subspaces $\text{Range}(\boldsymbol{W}_{r, \mathcal{L}}), \text{Range}(\boldsymbol{W}_{r, \mathcal{W}}) \in \text{Gr}(r,m)$ defining the interpolation. That is, the geodesic $[\boldsymbol{U}_r(s)]$ minimizes the distance between $\text{Range}(\boldsymbol{W}_{r, \mathcal{L}})$ and $\text{Range}(\boldsymbol{W}_{r, \mathcal{W}})$ while still constituting an $r$-dimensional subspace in $\mathbb{R}^m$.

\subsection{Ridge Optimization}
After approximating $\boldsymbol{W}_{r, \mathcal{L}}$ and $\boldsymbol{W}_{r, \mathcal{W}}$ we must make an informed decision to take the union of subspaces or compute a new subspace $\text{Range}(\boldsymbol{U}_r)$ against some criteria parametrized over the Grassmannian geodesic. Then we may restate the original problem with a common dimension reduction, $\boldsymbol{y} = \boldsymbol{U}_r^T\boldsymbol{\theta}$, utilizing updated approximations over $r < m$ combined/mixed active coordinates,
\begin{equation} \label{eq:r_multiopt}
    \underset{\boldsymbol{y} \in \mathcal{Y}}{\text{maximize}} \,\,tH_{\mathcal{L}}(\boldsymbol{y}; \boldsymbol{x}) + (1-t)H_{\mathcal{W}}(\boldsymbol{y}; \boldsymbol{x}),
\end{equation}
for all $t \in [0,1]$. Once again, this optimization problem involves a closed and bounded feasible domain of parameter values $\mathcal{Y} = \lbrace \boldsymbol{y} \in \mathbb{R}^r \,:\, \boldsymbol{y} = \boldsymbol{U}_r^T\boldsymbol{\theta},\,\, \boldsymbol{\theta}\in \mathcal{D}\rbrace$ which remains convex for convex $\mathcal{D}$ and a new subspace $\text{Range}(\boldsymbol{U}_r)$. The utility of the dimension reduction is the ability to formulate a continuous \emph{trace} of the Pareto front \cite{bolten2020tracing}---involving the inverse of a convex combination of Hessians---in fewer dimensions. This is supplemented by visualization in the case $r=1$ or $r=2$ providing \emph{empirical evidence of convexity}. The resulting convex approximations and visualizations are summarized in Section \ref{sec:Numerics}.

\subsection{Computational Considerations}
In order to approximate the eigenspaces of $\boldsymbol{C}_{\mathcal{L}}$ and $\boldsymbol{C}_{\mathcal{W}}$ for the separate responses \eqref{eq:throughputs} we must first approximate the gradients of the network throughput responses which are not available in an analytic form. Specifically, we use forward finite difference approximations and Monte Carlo as a quadrature to approximate the integral of partial derivatives in \eqref{eq:C}. These computations are supplemented by a rescaling of all parameters to a unit-less domain which permits consistent finite-difference step sizes.

The rescaling transformation is chosen based on the provided upper and lower bounds, summarized in Table \ref{Tb1}. This ensures that the scale of any one parameter does not influence  finite difference approximations. Moreover, this alleviates the need for an interpretation or justification when taking linear combinations of parameters with differing units. Because the throughput calculations involve parameter combinations appearing as exponents in the composition of a variety of computations, we use a uniform sampling of log-scaled parameter values. This transforms parameters appearing as exponents to appear as coefficients---a useful transformation given that we ultimately seek an approximation of linear combinations of parameters inherent to the definition of a subspace.

The resulting scaling of the domain is achieved by the composition of transformations
$\boldsymbol{\tilde{\theta}} = \boldsymbol{M}\ln(\boldsymbol{\theta}) + \boldsymbol{b}$
where $ \boldsymbol{M} = \text{diag}(2/(\ln(\theta_{u,1}) - \ln(\theta_{\ell,1})),\dots,2/(\ln(\theta_{u,m}) - \ln(\theta_{\ell,m})))$, $\boldsymbol{b} = -\boldsymbol{M}\mathbb{E}[\ln(\boldsymbol{\theta})]$, and $\ln(\cdot)$ is taken component-wise. To compute this transformation, we take $\theta_{\ell,i}$ and $\theta_{u,i}$ as the $i$-th entries of the lower and upper bounds and $\mathbb{E}[\ln(\boldsymbol{\theta})]$ the mean of $\ln(\boldsymbol{\theta}) \sim U_m[\ln(\boldsymbol{\theta}_{\ell}), \ln(\boldsymbol{\theta}_u)]$. This particular choice of scaling ensures $\boldsymbol{\tilde{\theta}}\in [-1, 1]^m$ and $\mathbb{E}[\boldsymbol{\tilde{\theta}}] = \boldsymbol{0}$ so the resulting domain is also centered. Lastly, we use Monte Carlo as a constant coefficient quadrature rule to approximate the integral form of the two separate matrices, defined by \eqref{eq:C}, for the two throughputs. The details are provided as Algorithm \ref{alg1}. 

The selection of $r$ in Algorithm \ref{alg1} can be automated by, for example, a heuristic which takes the largest gap in eigenvalues \cite{Constantine2015}. 
For simplicity, we take an exploratory approach to sele-

\begin{algorithm}[t] 
\captionsetup{font=footnotesize}
\caption{Monte Carlo Approximation of Throughput Active Subspaces using Forward Differences}
\footnotesize
    \begin{algorithmic}[1]\label{alg1}
	\REQUIRE Forward maps $S_{\mathcal{L}}$ and $S_{\mathcal{W}}$, small coordinate perturbation $h \geq 0$, fixed scenario parameters $\boldsymbol{x}$, and parameter bounds $\boldsymbol{\theta}_{\ell}$, $\boldsymbol{\theta}_{u} \in \mathbb{R}^m$.
    \STATE Generate $N$ random samples uniformly, $\lbrace \boldsymbol{\tilde{\theta}}_i\rbrace_{i=1}^N \sim U_m[-1,1]$.
    \STATE Compute $\boldsymbol{M}$, $\boldsymbol{M}^{-1}$, and $\boldsymbol{b}$ according to a uniform distribution of log-scale parameters given $\boldsymbol{\theta}_{\ell}$ and $\boldsymbol{\theta}_{u}$.
    \FOR{$i=1$ to $N$}
    \STATE Transform the uniform log-scale sample to the original scale $\boldsymbol{\theta}_i = \text{exp}(\boldsymbol{M}^{-1}(\boldsymbol{\tilde{\theta}}_i - \boldsymbol{b}))$ where the exponential is taken component-wise.
    \STATE Evaluate forward maps $(S_{\mathcal{L}})_i = S_{\mathcal{L}}(\boldsymbol{\theta}_i; \boldsymbol{x})$ and $(S_{\mathcal{W}})_i = S_{\mathcal{W}}(\boldsymbol{\theta_i}; \boldsymbol{x})$.
    \FOR {$j=1$ to $m$}
    \STATE Transform the $j$-th coordinate perturbation to the original input scale
    $\boldsymbol{\theta}_h = \text{exp}(\boldsymbol{M}^{-1}(\boldsymbol{\tilde{\theta}}_i + h\boldsymbol{e}_j - \boldsymbol{b}))$.
    \STATE Approximate the $j$-th entry of the gradient at the $i$-th sample as
    $$
    (\tilde{\nabla} S_{\mathcal{L}})_{i,j} = \frac{S_{\mathcal{L}}(\boldsymbol{\theta}_h; \boldsymbol{x}) - (S_{\mathcal{L}})_i}{h},
    $$
    similarly for $(\tilde{\nabla} S_{\mathcal{W}})_{i,j}$, where $\boldsymbol{e}_j$ is the $j$-th column of the $m$-by-$m$ identity matrix.
    \ENDFOR
    \ENDFOR
    \STATE Take the average of the outer product of approximated gradients as
    $$
    \boldsymbol{\tilde{C}}_{\mathcal{L}} = \frac{1}{N}\sum_{i=1}^N (\tilde{\nabla} S_{\mathcal{L}})_{i,;} \otimes (\tilde{\nabla} S_{\mathcal{L}})_{i,;},
    $$
    similarly for $\boldsymbol{\tilde{C}}_{\mathcal{W}}$, where the tensor (outer) product is taken over the $j$-th index.
    \STATE Approximate the eigenvalue decompositions $$\boldsymbol{\tilde{C}}_{\mathcal{L}} = \boldsymbol{\tilde{W}}_{\mathcal{L}}\boldsymbol{\tilde{\Lambda}}_{\mathcal{L}}\boldsymbol{\tilde{W}}_{\mathcal{L}}^T \quad \text{and}\quad \boldsymbol{\tilde{C}}_{\mathcal{W}} = \boldsymbol{\tilde{W}}_{\mathcal{W}}\boldsymbol{\tilde{\Lambda}}_{\mathcal{W}}\boldsymbol{\tilde{W}}_{\mathcal{W}}^T$$
    ordered by decreasing eigenvalues.
    \STATE Observe the eigenvalue decay and associated gaps to inform a reasonable choice of $r$.
    \RETURN The first $r$ columns of $\boldsymbol{\tilde{W}}_{\mathcal{L}}$ and $\boldsymbol{\tilde{W}}_{\mathcal{W}}$, denoted $\boldsymbol{\tilde{W}}_{r,\mathcal{L}}$ and $\boldsymbol{\tilde{W}}_{r,\mathcal{W}}$.
    \end{algorithmic}
\end{algorithm}
\normalsize
%
\noindent
 -cting $r$ which requires some user-input. We seek a visualization of the response to provide empirical evidence that the throughputs are predominantly convex and hence require $r\leq 2$. We then check that the result offers acceptable approximations of throughputs with sufficient gaps in the second and third eigenvalues suggesting reasonable subspace approximations.

\section{Simulation and Results} \label{sec:Numerics}

We demonstrate the ideas proposed in Section \ref{sec:AS} on the LAA-Wi-Fi coexistence scenario described in Section \ref{sec:sys_model} to maximize both throughputs simultaneously. We apply active subspaces to simplify the multi-criteria optimization problem \eqref{eq:multiopt} by focusing on a reduced set of mixed PHY and MAC layer parameter combinations informing a trace of near Pareto optimal solutions. Table \ref{Tb1} summarizes the scenario parameters and parameter bounds used to inform the throughput computations and domain scaling, respectively. 

\begin{figure*}[ht]
\vspace{0.01in}
\begin{subfigure}{.5\textwidth}
  \centering
  \includegraphics[width=.75\linewidth]{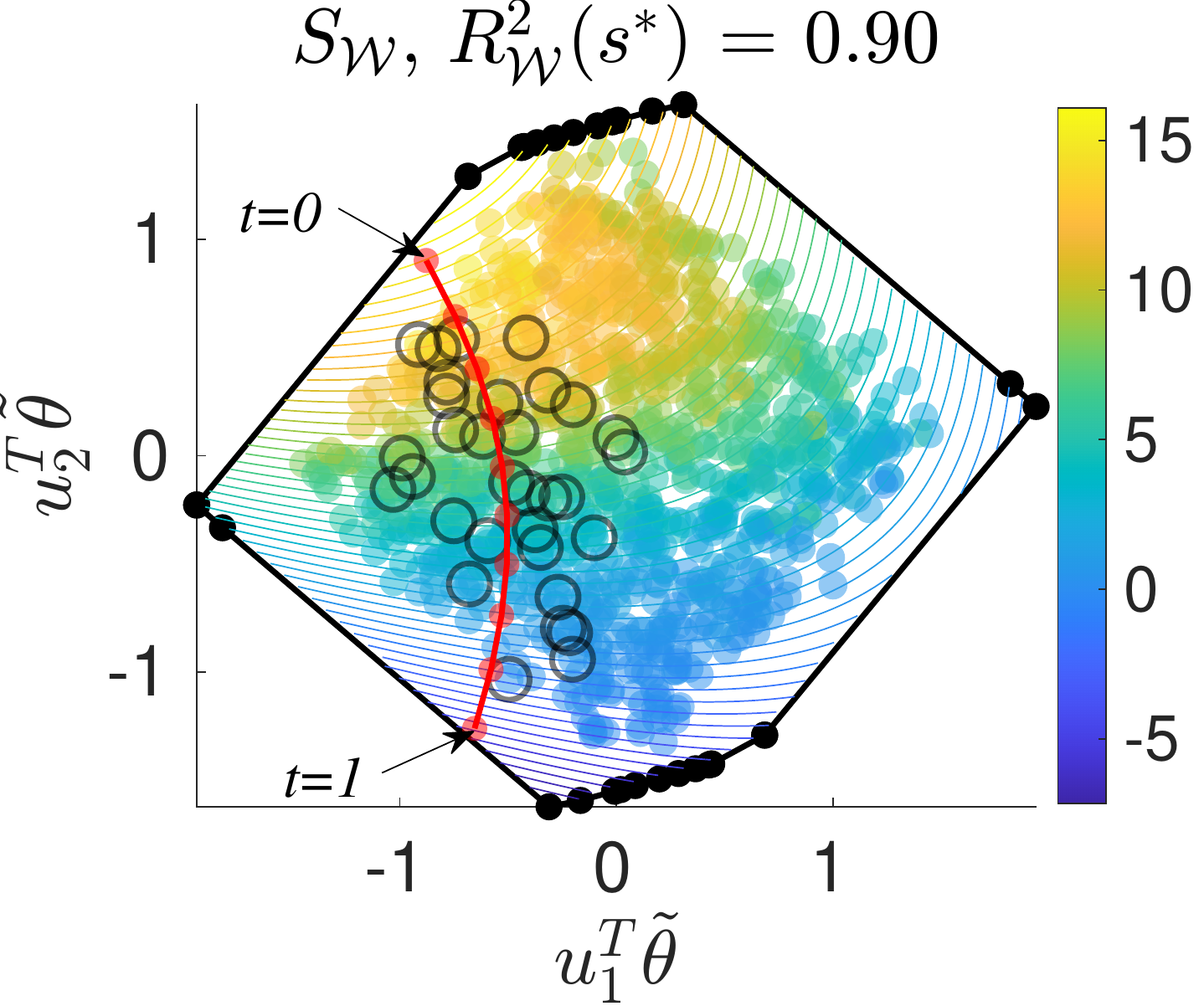}
  \label{fig:sub-first}
  \hfill
  \caption{Wi-Fi Throughput Shadow Plot}
\end{subfigure}
\begin{subfigure}{.5\textwidth}
  \centering
  \includegraphics[width=.75\linewidth]{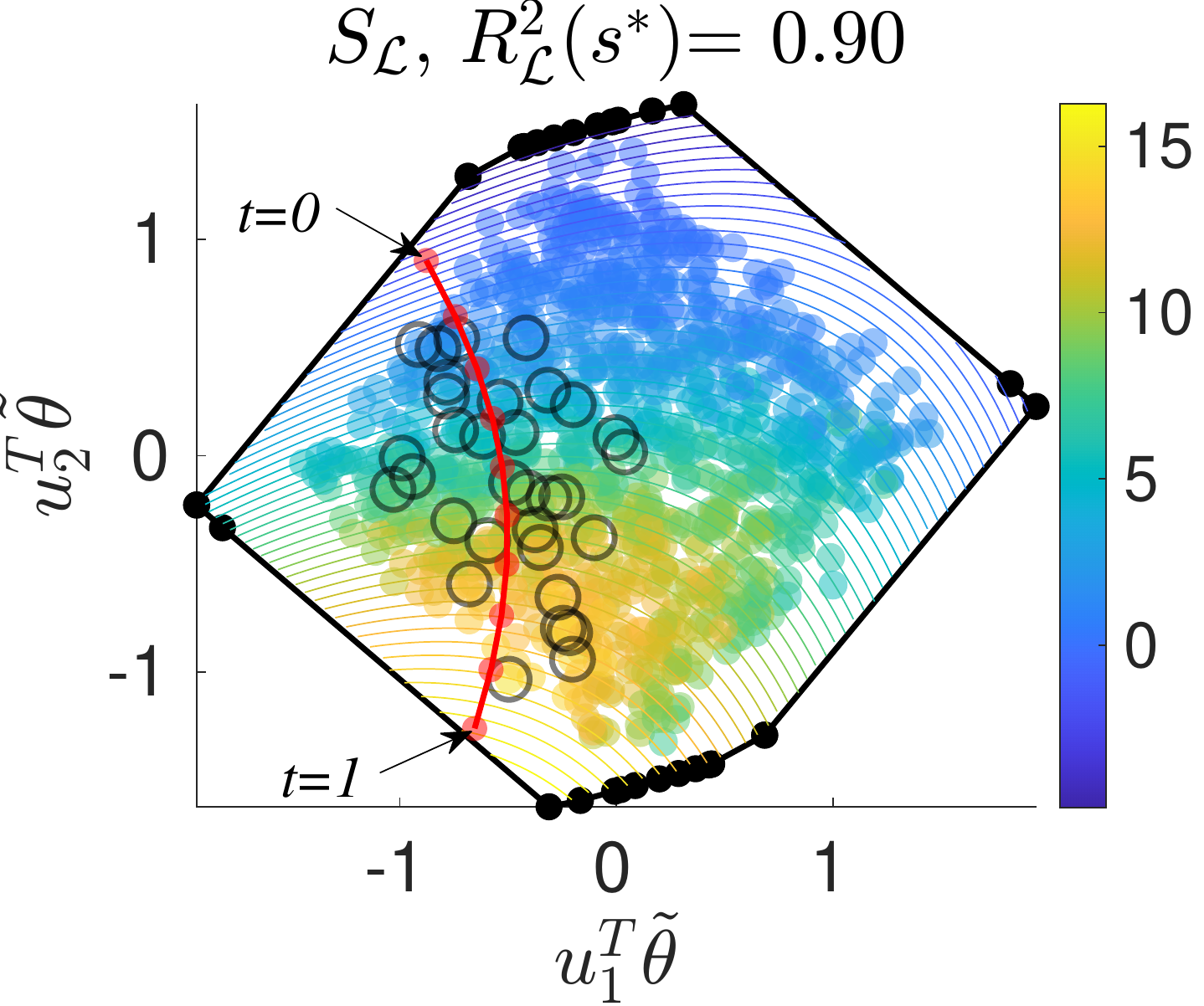}
  \caption{LAA Throughput Shadow Plot}
  \label{fig:sub-second}
\end{subfigure}
\caption{Pareto trace of quadratic ridge profiles. The quadratic Pareto trace (red curve and dots) is overlaid on a shadow plot over the mixed coordinates (colored scatter) with the projected bounds and vertices of the domain (black dots and lines), $\mathcal{Y}$. The quadratic approximations (colored contours) are computed as least-squares fits over the mixed subspace coordinates. Also depicted is the projection of the non-dominated domain values from the $N=1000$ random samples (black circles). The trace begins at $t=0$ with near maximum quadratic Wi-Fi throughput and we move (smoothly) along the red curve to $t=1$ obtaining near maximum quadratic LAA throughput---maintaining an approximately best trade-off over the entire curve restricted to $\mathcal{Y}$.}
\label{fig:fig_1}
\end{figure*}
\normalsize

\begin{figure}[ht]
\centerline{\includegraphics[width=1\linewidth]{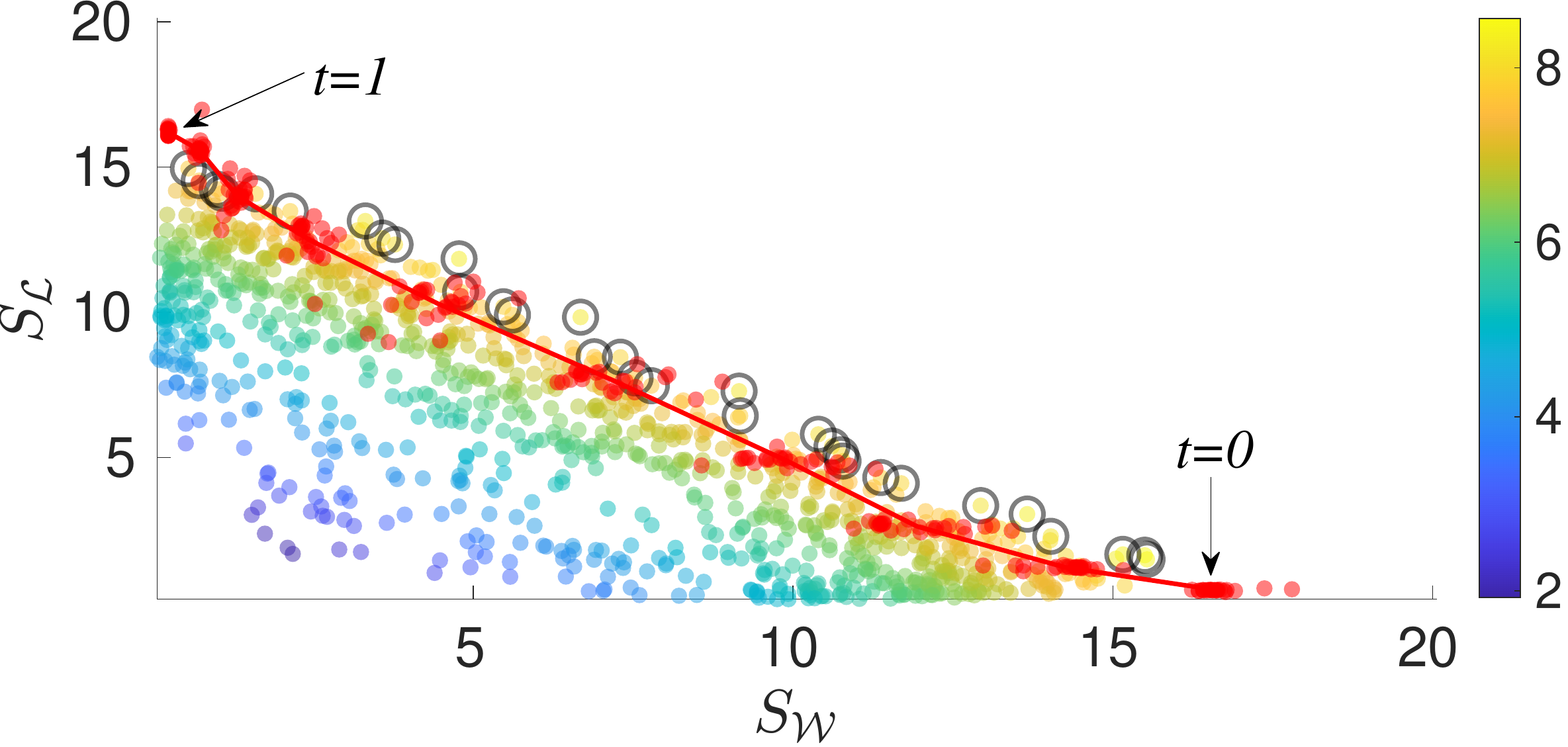}}
\caption{Approximation of the Pareto front resulting from the quadratic trace. The approximate Pareto front (red curve) is shown with the non-dominated throughput values (black circles) and scatter of $N = 1000$ random responses colored according to the averaged throughputs. The red curve is the image of the continuous trace through parameter space (visualized as the red curve in Fig. \ref{fig:fig_1}) representing a near Pareto optimal set of solutions.}
\label{fig_2}
\end{figure}
\normalsize
\raggedbottom

The numerical experiment utilizes $N = 1000$ samples resulting in $N(m+1) = 18,000$ function evaluations to compute forward differences with $h = 10^{-6}$. The resulting decay in eigenvalues gave relatively accurate degree-$2$ to degree-$5$ least-squares polynomial approximations---computed utilizing sets of coordinate-output pairs $\lbrace (\boldsymbol{\tilde{W}}_{r,\mathcal{L}}^T\boldsymbol{\tilde{\theta}}_i, (S_{\mathcal{L}})_i)\rbrace_{i=1}^N$ and $\lbrace (\boldsymbol{\tilde{W}}_{r,\mathcal{W}}^T\boldsymbol{\tilde{\theta}}_i, (S_{\mathcal{W}})_i)\rbrace_{i=1}^N$---with varying coefficients of determination between $0.83-0.98$ for both throughputs when $r=1$ or $r=2$. In an effort to improve the ridge approximation while retaining the ability to visualize the response for the quadratic polynomial, we fix $r=2$ and mix the subspaces according to a quadratic approximation with corresponding coefficients of determination $R^2_{\mathcal{L}}$ and $R^2_{\mathcal{W}}$---computed using new sets of subspace coordinates and throughput (coordinate-output) pairs. We select a criteria to mix subspaces achieving a balanced approximation when $r=2$. This offers a subproblem,
\begin{equation} \label{eq:sub_multi}
    \underset{s \in [0,1]}{\text{maximize}}\,\,\min\lbrace R^2_{\mathcal{L}}(s), R^2_{\mathcal{W}}(s)\rbrace,
\end{equation}
where the separate throughput coefficients of determination, $R^2_{\mathcal{L}}(s)$ and $R^2_{\mathcal{W}}(s)$, are parametrized by successive quadratic fits over new coordinates $\boldsymbol{y}_i = \boldsymbol{U}^T_r(s)\boldsymbol{\tilde{\theta}}_i$ for all $i=1,\dots,N$ defined by the Grassmannian geodesic $[\boldsymbol{U}_r(s)]$ beginning at $\text{Range}(\boldsymbol{\tilde{W}}_{r,\mathcal{W}})$ and ending at $\text{Range}(\boldsymbol{\tilde{W}}_{r,\mathcal{L}})$. 

The univariate subproblem in \eqref{eq:sub_multi} can be visualized and, in this experiment, achieved a unique maximizing argument $s^* \in [0,1]$ admitting a mixed subspace with orthonormal basis given by two columns in a matrix $\boldsymbol{U}_r(s^*) = [\boldsymbol{u}_1 \,\, \boldsymbol{u}_2]$, i.e., $ [\boldsymbol{u}_1 \,\, \boldsymbol{u}_2] \in \mathbb{R}^{m\times 2}$ taken at the optimal $s^*$ is the representative element of $[\boldsymbol{U}_r(s^*)]$. The coefficients of determination varied monotonically and intersected over the Grassmannian parametrization. Consequently, the subproblem results in an approximately equal criteria for the accuracy of the quadratic ridge profiles $H_{\mathcal{W}}$ and $H_{\mathcal{L}}$, i.e., $R^2_{\mathcal{L}}(s^*) \approx R^2_{\mathcal{W}}(s^*) \approx 0.90$. The choice of quadratic least-squares approximation over mixed active coordinates admits an analytic form \eqref{eq:quad_trace} for the Pareto trace of \eqref{eq:r_multiopt}. The analytic form of the quadratic trace is taken as a convenience in contrast to a higher-order polynomial approximation---or alternative approximation---and subsequent trace. Part of the utility afforded by the dimension reduction is simultaneously fitting and visualizing higher-order approximations over the low-dimensional coordinates for fixed $N$ random samples of coordinate-output pairs. However, the quadratic fits were deemed reasonable approximations admitting convexity which can be observed directly in the shadow plots. The convex quadratic ridge approximations, quadratic Pareto trace, and projected boundary of the domain over the mixed subspace are shown in Fig. \ref{fig:fig_1}. Additionally, the Pareto front approximation resulting from the quadratic trace is shown with the non-dominated designs in Fig. \ref{fig_2}.

Observing Fig. \ref{fig:fig_1}, the continuous Pareto trace over the subspace coordinates (red curve) moves approximately through the collection of projected non-dominated designs (black circles). The non-dominated designs are determined from the $N=1000$ random samples; sorted according to \cite{kung1975}. However, it is not immediately clear through this visualization that the non-dominated designs constitute elements of an alternative continuous approximation of the Pareto front---perhaps represented by an alternative low-dimensional manifold. Instead, we have supplemented a continuous parametrization of the Pareto front which is implicitly regularized as a solution over a low-dimension subspace. However, there are infinite $\boldsymbol{\theta}$ in the original parameter space which correspond to points along the trace depicted in Fig. \ref{fig:fig_1}---i.e., infinitely many $m-r$ inactive coordinate values which may change throughputs albeit significantly less than the two mixed active coordinates, $y_1=\boldsymbol{u}_1^T\boldsymbol{\tilde{\theta}}$ and $y_2=\boldsymbol{u}_2^T\boldsymbol{\tilde{\theta}}$. To reconcile the choice of infinitely many inactive coordinates, we visualize subsets of 25 inactive coordinate samples drawn randomly over $\text{Null}(\boldsymbol{U}^T_r(s^*))$ along a discretization of the trace. Fig. \ref{fig_2} depicts the corresponding throughput evaluations from the inactive samples as red dots along the approximated Pareto front---the red line connects conditional averages of throughputs over inactive samples along corresponding points over the trace. The visualization emphasizes that the throughputs change significantly less over the inactive coordinates in contrast to the range of values observed over the trace. 

There is some bias in the approximation of the Pareto front (red curve) in Fig. \ref{fig_2} which is not a least-squares curve of non-dominated throughput values (black circles) potentially due in part to the quadratic ridge approximations. We expect refinements to these approximations will further improve the continuous Pareto approximation (shown in red in Fig. \ref{fig_2}).
\raggedbottom

\section{Conclusion \& Future Work} 


We have proposed a technique to simultaneously optimize the performance of multiple MNOs sharing a single spectrum resource. An exploratory analysis utilizing an example of LTE-LAA coexistence with Wi-Fi network identified a common subspace-based dimension reduction of a basic network-behavior model. This enabled visualizations and low-dimensional approximations which led to a \emph{continuous} approximation of the Pareto frontier for the multi-criteria problem of maximizing all convex combinations of network throughputs over MAC and PHY parameters. Such a result simplifies the search for parameters which enable high quality performance of both networks, particularly compared to approaches which do not operate on a reduced parameter space. Analysis of the LAA-Wi-Fi example revealed an explainable and interpretable solution to an otherwise challenging problem---\emph{devoid of any known convexity or degeneracy until subsequent exploration}. 

Future work will incorporate alternative low-dimensional approximations including both cases of Grassmannian mixing and subspace unions to improve the trace. We will also summarize a sensitivity analysis, active subspace approximation diagnostics, and the parametrization of a predominantly flat manifold of near Pareto optimal solutions.

Implementations of this work are anticipated to enable spectrum sharing in unlicensed bands by simplifying the design of wireless network operation (control) and architecture. In on-going work, we are investigating methods to estimate the set of PHY and MAC layer parameters which lead to acceptable values of KPIs for multiple coexisting wireless networks. Principled approaches for quantifying this novel concept, coined the region of wireless coexistence (RWC), can lead to significantly-improved network operation. Real-world RWCs are high-dimensional sets with tens or hundreds of parameters, and existing models are designed for low-dimensional problems. We aim to use the dimensionality-reduction techniques described above to combat issues with high dimensionality. We hope to additionally accelerate existing RWC algorithms by efficiently taking advantage of a parameter-manifold of near-optimal solutions with reduced intrinsic dimension by applying the methods described in this work. We anticipate that new methods which leverage compositions with the presented dimension reduction will benefit from accelerations and regularization in an otherwise challenging high-dimensional formulation.

\bibliographystyle{IEEEtran}
\balance
\bibliography{bibliography}

\end{document}